\documentclass[letterpaper, 10 pt, conference]{ieeeconf}  % Comment this line out if you need a4paper

\IEEEoverridecommandlockouts                              % This command is only needed if 
\usepackage{graphics} % for pdf, bitmapped graphics files
\usepackage{amsmath} % assumes amsmath package installed
\usepackage{amssymb}  % assumes amsmath package installed
\usepackage{enumerate}
\usepackage{mathtools}
\usepackage[noadjust]{cite}
\usepackage[demo]{graphicx}
\usepackage{caption}
\usepackage{subcaption}

\usepackage{tikz}
\usetikzlibrary{patterns,decorations.pathmorphing,decorations.markings,calc}

\newtheorem{thm}{\textbf{Theorem}}
\newtheorem{lem}{\textbf{Lemma}}
\newtheorem{prop}{\textbf{Proposition}}
\newtheorem{cor}{\textbf{Corollary}}
\newtheorem{defn}{\textbf{Definition}}

\newtheorem{rem}{\textbf{Remark}}
\newtheorem{assum}{\textbf{Assumption}}

\newcommand{\bm}[1]{\boldsymbol{#1}}
\DeclareMathOperator{\diag}{diag}
\DeclareMathOperator{\sgn}{sgn}

\title{\LARGE \bf
Learning High-Order Control Barrier Functions\\ for Safety-Critical Control with Gaussian Processes
}

\author{Mohammad Aali and Jun Liu% <-this % stops a space
\thanks{This work was supported in part by Nutrien Ltd., NSERC of Canada, and the Canada Research Chairs program.}
\thanks{Mohammad Aali is with the Department of Applied Mathematics, University of Waterloo, Canada,
        {\tt\small mohammad.aali@uwaterloo.ca}}%
\thanks{Jun Liu is with the Department of Applied Mathematics, University of Waterloo, Canada,
        {\tt\small j.liu@uwaterloo.ca}}%
}

\begin{document}

\maketitle
\thispagestyle{empty}
\pagestyle{empty}

%%%%%%%%%%%%%%%%%%%%%%%%%%%%%%%%%%%%%%%%%%%%%%%%%%%%%%%%%%%%%%%%%%%%%%%%%%%%%%%%
\begin{abstract}
Control barrier functions (CBFs) have recently introduced a systematic tool to ensure system safety by establishing set invariance. When combined with a nominal control strategy, they form a safety-critical control mechanism. However, the effectiveness of CBFs is closely tied to the system model. In practice, model uncertainty can compromise safety guarantees and may lead to conservative safety constraints, or conversely, allow the system to operate in unsafe regions. In this paper, we use Gaussian processes to mitigate the adverse effects of uncertainty on high-order CBFs (HOCBFs). A particular structure of the covariance function enables us to convert the chance constraints of HOCBFs into a second-order cone constraint, which results in a convex constrained optimization as a safety filter. We analyze the feasibility of the resulting optimization and provide the necessary and sufficient conditions for feasibility. The effectiveness of the proposed strategy is validated through two numerical results.
\end{abstract}
%%%%%%%%%%%%%%%%%%%%%%%%%%%%%%%%%%%%%%%%%%%%%%%%%%%%%%%%%%%%%%%%%%%%%%%%%%%%%%%%
\section{INTRODUCTION} \label{sec1}

In recent years, the requirements of modern control applications have extended beyond mere stability assurance to include safety verification. The notion of safety holds such major importance that we deploy an autonomous system in the real world only if it is verified to be safe or achieves a high level of probabilistic safety. CBFs have recently provided a powerful theoretical tool for synthesizing controllers that ensure the safety of dynamical systems. They were initially applied to autonomous driving and bipedal walking \cite{ames2019control}. While this method has shown significant promise in addressing safety verification challenges, it had an initial limitation: it was only suitable for systems of relative degree one with respect to the CBF. The introduction of exponential CBFs \cite{nguyen2016exponential} eliminated this limitation and expanded the application of this approach to systems of arbitrary relative degree. Consequently, the use of this safety method has extended to a wide range of robotic systems. A more general definition of HOCBFs is proposed in \cite{xiao2021high}. More recently, the assumptions of forward completeness and uniform relative degree have been relaxed for HOCBFs \cite{tan2021high}.

However, the CBFs method initially advanced on the theoretical assumption that a precise dynamical model of the system is available. In practice, model uncertainty compromises the safety guarantees provided by CBFs. Data-driven methods offer powerful tools to address these challenges with adaptive or probabilistic solutions. Among machine learning techniques, Gaussian processes (GPs) have gained significant attention in control systems. Beyond prediction, GPs provide probabilistic confidence intervals, which are invaluable in data-driven safety-critical control methods where uncertainty quantification is crucial.

Data-driven methods have been extensively used to ensure safety in the presence of model uncertainty with CBFs of relative degree one. In \cite{wang2018safe}, a GP is employed to learn the residual control-independent dynamics for the safe navigation of quadrotors with CBFs. A GP-based systematic controller synthesis method is introduced in \cite{jagtap2020control}, offering an inherently safe control strategy. In \cite{fan2020bayesian}, the authors employ a Bayesian approach to learn model uncertainty while ensuring safety during the learning process. All the aforementioned works assume that the actuation term is known.

Data-driven methods for HOCBFs can be divided into two approaches. The first approach learns the unknown (residual) dynamics and uses the approximated model to derive safety certificates. The impact of the residual dynamics on ECBFs is addressed in \cite{wang2021learning}, where system dynamics are learned using a neural network. Methods employing neural networks must achieve precise approximations of the dynamics to formally ensure safety. In general, this requirement is challenging to meet, but GPs can provide a solution by offering confidence bounds. In \cite{khojasteh2020probabilistic}, completely unknown control-affine system dynamics are approximated with matrix variate GP (MVGP), resulting in a probabilistic high-order safety constraint. This method ensures that the expectation and variance of the resulting probabilistic safety constraint remain linear and quadratic in the control input, leading to convex optimization. However, approximating the entire system using GPs is computationally expensive. A prevalent assumption, therefore, is to learn each element of the vector fields individually and overlook their correlations.

These limitations lead us to the second approach, which approximates the projection of the residual of system dynamics onto the CBF. In \cite{taylor2020learning}, an episodic learning method is employed to approximate the impact of unmodeled dynamics on system safety, eliminating the need for exhaustive data collection. In \cite{castaneda2021gaussian}, a GP-based min-norm controller stabilizing an unknown control affine system using control Lyapunov functions is introduced. The effect of the uncertainty in the stability constraint is approximated by a GP model. The same method is applied to a relative degree one CBF design in \cite{castaneda2021pointwise}. This technique results in the approximation of a real-valued function representing the influence of uncertainty on the safety certificate.

In this paper, we present a data-driven approach based on GP that learns the impact of model uncertainty on high-order safety certificates. We then incorporate the uncertainty-aware high-order safety constraint with an arbitrary robust controller in a min-norm quadratic program. The main contributions of this paper relative to prior works are as follows:
\begin{itemize}
\item We characterize the impact of uncertainty on the HOCBF. It is proven that with a proper design of HOCBFs, the residual term in the high-order safety certificate is control affine, which facilitates obtaining the GP-based high-order safety constraint.
\item We demonstrate that by selecting a particular form of the covariance function in the GP structure, the resulting min-norm QP with an uncertainty-aware safety constraint can be converted into a second-order cone program (SOCP).
\item We analyze the feasibility of the resulting constrained optimization problem and determine the necessary and sufficient conditions for feasibility.
\end{itemize}
% A conference version of this paper will appear in \cite{aali2024}. This extended version includes additional details omitted in the conference paper.
\section{BACKGROUND AND PRELIMINARIES}

Consider the nonlinear control affine system
\begin{equation}
    \dot {\boldsymbol{x}} = f(\boldsymbol{x})+ g(\boldsymbol{x}) \boldsymbol{u},
    \label{eq1}
\end{equation}where $\boldsymbol{x} \in \mathbb{R}^n$is the state, $\boldsymbol{u} \in \mathbb{R}^m$  is the control input, and the vector fields $f:\mathbb{R}^n \rightarrow \mathbb{R}^n$ and $g:\mathbb{R}^n \rightarrow \mathbb{R}^{n \times m}$, are locally Lipschitz in $\boldsymbol{x}$. Given a locally Lipschitz continuous state feedback control $\pi:\mathbb{R}^n \rightarrow \mathbb{R}^{m}$, the closed-loop system is:
\begin{equation}
    \dot {\boldsymbol{x}} = F_{cl}(\boldsymbol{x}) \triangleq  f(\boldsymbol{x})+ g(\boldsymbol{x}) \pi(\boldsymbol{x}),
    \label{eq2}
\end{equation}
 which is locally Lipschitz continuous. Then, for any initial state $\boldsymbol{x}(0) = \boldsymbol{x}_0$, the system (\ref{eq2}) has a unique solution $\boldsymbol{x}(t)$ defined on a maximal interval of existence $I(\boldsymbol{x}_0) = [0, I_{max})$. In this paper, we consider forward complete closed-loop systems, thus $I_{max} = \infty$. A set $\mathcal{C} \subset \mathbb{R}^n$ is said to be forward invariant for (\ref{eq2}), if for all $x(0) \in \mathcal{C}$, the solution $\boldsymbol{x}(t) \in \mathcal{C}$ for all $t \in I(\boldsymbol{x}_0)$.

\begin{defn}[Class $\mathcal{K}$ function \cite{khalil2002nonlinear}]\label{def0}
    We say that a continuous function $\alpha:[0, a) \rightarrow [0, \infty), a>0$, belongs to class $\mathcal{K}$, if it is strictly increasing and $\alpha(0)=0$.
\end{defn}

CBF method defines safety based on the notion of set invariance, where a subset of the state space is specified as the safe set. This set is characterized by the zero-superlevel set of a continuously differentiable function $h:\mathbb{R}^n \rightarrow \mathbb{R}$ as
\begin{align}
    \mathcal{C} &= \{\boldsymbol{x} \in \mathbb{R}^n \mid h(\boldsymbol{x})\geq 0\},\nonumber\\
    \partial \mathcal{C} &= \{\boldsymbol{x} \in \mathbb{R}^n \mid h(\boldsymbol{x})= 0\}.\label{eq3}
\end{align}

CBFs provide a constructive tool for achieving forward invariance of set $\mathcal{C}$. We define a CBF as follows:
\begin{defn}[Control barrier function \cite{ames2019control}]\label{def2}
     Given the set $\mathcal{C}$ as defined in (\ref{eq3}), the continuously differentiable function $h(\boldsymbol{x})$ is called a control barrier function on a domain $\mathcal{D}$ with $\mathcal{C} \subset \mathcal{D} \subset \mathbb{R}^n$, if there exists a class $\mathcal{K}$ function
    $\alpha$ such that 
\begin{equation}
    \sup_{\boldsymbol{u} \in \mathbb{R}^m} [L_f h(\boldsymbol{x})+ L_g h(\boldsymbol{x})\boldsymbol{u}] +\alpha(h(\boldsymbol{x}))  \geq 0,
    \label{eq4}
\end{equation}for all $\boldsymbol{x} \in \mathcal{D}$, where $L_f{h(\boldsymbol{x})}$ and $L_g h(\boldsymbol{x})$ are the Lie derivatives of $h(\boldsymbol{x})$ with respect to the vector fields $f$ and $g$.
\end{defn}
We can derive the following corollary based on Nagumo's theorem \cite{blanchini2008set}, for the forward invariance of the safe set $\mathcal{C}$.
\begin{cor}\label{cor1}
     Given CBF $h:\mathbb{R}^n \rightarrow \mathbb{R}$ with the associated set $\mathcal{C}$ in (\ref{eq3}), if $\nabla h(\bm x) \neq 0$ for all $\bm x \in \partial \mathcal{C}$, any Lipschitz continuous controller $\boldsymbol{u}(\boldsymbol{x})$ satisfying (\ref{eq4}) guarantees that $\mathcal{C}$ is forward invariant for the system (\ref{eq1}) and thus safe.
\end{cor}

\subsection{High-order control barrier functions}\label{1subsec1}

In the context of high-order control barrier functions, first, we introduce relative degree:
\begin{defn}[Relative degree]\label{de3}
     A continuously differentiable function $h$ is said to have relative degree $r$ on a given domain $\mathcal{D}$ with respect to system (\ref{eq1}) if:
    \begin{enumerate}[(i)]
        \item $L_g L_f^k h(\boldsymbol{x}) = 0$ for all $k < r-1$,
        \item $L_g L_f^{r-1} h(\boldsymbol{x}) \neq 0$,
    \end{enumerate}
   hold for all $\bm x \in \mathcal{D}$.
\end{defn}
While Definition \ref{def2} is only applicable to CBFs with a relative degree of one, many applications often involve CBFs with higher relative degrees. To accommodate such scenarios, an extended definition known as HOCBFs has been developed in \cite{xiao2021high}.
In this approach, we define a series of continuously differentiable functions $\zeta:\mathbb{R}^n \rightarrow \mathbb{R}$ as
\begin{align}
\zeta_0(\boldsymbol{x}) ={}& h(\boldsymbol{x}),\nonumber\\
\zeta_{i}(\boldsymbol{x}) ={}& \dot \zeta_{i-1}(\boldsymbol{x}) + \alpha_i(\zeta_{i-1}(\boldsymbol{x})),\; i \in \{1, \dots, r \},
\label{eq5}
\end{align}
where $\alpha_i(\cdot),\, i \in \{1, \dots, r \}$ are class $\mathcal{K}$ functions. We also define their zero-superlevel sets $\mathcal{Z}_i$ and their interior sets for $i \in \{1, \dots, r \}$ as 
\begin{align}
\mathcal{Z}_i ={}& \{\boldsymbol{x} \in \mathbb{R}^n \mid \zeta_{i-1}(\boldsymbol{x}) \geq 0\},\label{eq6}\\
Int(\mathcal{Z}_i) ={}& \{\boldsymbol{x} \in \mathbb{R}^n \mid \zeta_{i-1}(\boldsymbol{x}) > 0\}.
\label{eq7}
\end{align}

\begin{defn}\label{def3}
Let the functions $\zeta_i(\boldsymbol{x})$ and sets $Int(\mathcal{Z}_i)$ be defined by (\ref{eq5}) and (\ref{eq7}), respectively. The $r^{th}$ order continuously differentiable function $h(\boldsymbol{x})$ with relative degree $r>1$ is called a HOCBF if $h$ and its derivatives up to order $r$, are locally Lipschitz continuous, and there exists a set of sufficiently smooth class $\mathcal{K}$ functions $\alpha_i$, such that 

\begin{align}
\sup_{u \in \mathbb{R}^m}[L_f^r h+ L_gL_f^{r- 1} h \hspace{2pt} \bm u &+ \sum _{i = 1}^{r - 1}L_f^i(\alpha _{r - i}\circ \zeta _{r-i-1})\nonumber\\
&+\alpha_r(\zeta_{r-1})] \geq 0,
\label{eq8}
\end{align}
for all $\bm x \in \bigcap^{r}_{i=1} Int(\mathcal{Z}_i)$, where the dependence on $\bm x$ is removed for simplicity.
\end{defn}
Note that (\ref{eq8}) is obtained by substituting the expression $\zeta_0(\bm x) = h(\bm x)$ into the recursive formula (\ref{eq5}) $r$ times. Therefore, it is equivalent to $\zeta_r(\bm x, \bm u) \geq 0$, which will be used in the remainder of the paper.
\begin{prop}\label{prop1}
Given system (\ref{eq1}) with a HOCBF $h$ with relative degree $r$, any Lipschitz continuous control $\boldsymbol{u}$ that satisfies (\ref{eq8}) renders the set $\bigcap^{r}_{i=1} Int(\mathcal{Z}_i)$ forward invariant for system (\ref{eq1}).
\end{prop}
A proof of this result can be found in \cite{aali2022multiple}. 

Inequality (\ref{eq8}) imposes an affine condition on the control values which can be used to ensure safety. Consider a nominal control policy $\bm u_{nom}$, Lipschitz continuous in $\bm x$, designed primarily to fulfill control objectives. Our goal is to apply this control law to the system only if it complies with (\ref{eq8}). In practice, this is accomplished by solving the following quadratic programming optimization:
\begin{align}
    \boldsymbol{u}_{s} ={} &\underset{\boldsymbol{u} \in \mathcal{R}^m}{\arg\min} \hspace{4pt} {\| \bm u - \bm u_{nom} \|^2_2} \nonumber\\
    &\textrm{s.t.} \quad \zeta_r(\bm x, \bm u) \geq 0.
    \label{eq9}
\end{align}
The resulting minimally invasive point-wise controller prioritizes the safety over control objectives by satisfying $\zeta_r(\bm x, \bm u) \geq 0$ as a hard constraint.

\section{Impact of model uncertainty on hocbfs}\label{sec3}
In this section, we focus on reformulating (\ref{eq8}) such that it takes the model uncertainty into account. We refer to the system (\ref{eq1}) as the true model, which is partially unknown, and consider a nominal model \begin{equation}
    \dot {\boldsymbol{x}} = \hat f(\boldsymbol{x})+ \hat g(\boldsymbol{x}) \boldsymbol{u},
    \label{eq10}
\end{equation}
where $\hat f:\mathbb{R}^n \rightarrow \mathbb{R}^n$ and $\hat g:\mathbb{R}^n \rightarrow \mathbb{R}^{n \times m}$ are known functions, locally Lipschitz in $\boldsymbol{x}$. We design a HOCBF based on the known nominal model (\ref{eq10}). To apply it to the true system, we make the following assumption.
\begin{assum}\label{assum1}
    If an $r^{th}$ order differentiable function $h:\mathbb{R}^n \rightarrow \mathbb{R}$ is an HOCBF for the nominal model (\ref{eq10}), it will be also a valid HOCBF for the true system (\ref{eq1}).
\end{assum}
Intuitively, this assumption requires the true and nominal model to share the same degree of actuation, which is met by feedback linearizable systems \cite{khalil2002nonlinear}.

Due to the inherent model mismatch between the true system and the nominal model, the Lie derivatives of $h(\bm x)$ (of order $r$) calculated based on the nominal model differ from their true values. We represent the resulting error by
\begin{align}
    \Delta_{i}(\bm x) &=  L^i_{f} h(\bm x) - L^i_{\hat f} h(\bm x), i \in \{1, \dots, r \},\label{eq11}\\
    \Delta_g(\bm x) &= L_gL_f^{r-1} h(\bm x) - L_{\hat g}L^{r-1}_{\hat f} h(\bm x).\label{eq12}
\end{align}
Similarly, the errors in the time derivatives of $h$ up to order $r-1$ are denoted by $\Delta_i(\bm x)$, where $i \in \{1, \dots, r-1\}$, while the error in the $r^{th}$ order time derivative is obtained by $\Delta_r(\bm x) + \Delta_g(\bm x) \hspace{2pt} \bm u$.

Consequently, these discrepancies propagate into the auxiliary functions (\ref{eq5}), potentially compromising the safety of the system. In this paper, our goal is to approximate the adverse effect of uncertainty on the high-order safety certificate (\ref{eq8}) by adopting a supervised learning approach. This approach makes our HOCBF design robust to model uncertainty. To achieve this, we must separate the propagated error in (\ref{eq8}) from the approximations based on the nominal model. This goal depends on the design of the HOCBF, which is addressed in the following result.

\begin{thm}\label{thm1}
    The true auxiliary functions in (\ref{eq5}) can be written as $\zeta_i = \hat{\zeta}_i + R_i\left(\Delta_i,\dots, \Delta_1\right)$, for $i \in \{1, \dots, r-1 \}$, where $\hat{\zeta}_i$ calculated based on the nominal model (\ref{eq10}) and $R_i\left(\Delta_i,\dots, \Delta_1\right)$ is a linear function of its argument,
    if and only if class $\mathcal{K}$ functions $\alpha_i(x) = k_i x$, where $k_i>0$. Furthermore, $R_i\left(\Delta_i,\dots, \Delta_1\right), i \in \{1, \dots, r-1 \}$ is obtained by
    \begin{align}
        R_i\left(\Delta_i,\dots, \Delta_1\right) &= \Delta_i + \sum_{1\leq j_1 \leq i}k_{j_1}\Delta_{i-1}+ \nonumber\\
        & \sum_{1\leq j_1 < j_2\leq i}{k_{j_1}k_{j_2}} \Delta_{i-2}+\dots+ \nonumber\\
        & \sum_{1\leq j_1<\dots<j_{i-1}\leq i}{k_{j_1}k_{j_2}\dots k_{j_{i-1}}} \Delta_1.
        \label{eq13}
    \end{align}
\end{thm}
\begin{proof}
    For the forward direction, we are given that $\zeta_i = \hat{\zeta}_i + R_i$, for $i \in \{1, \dots, r-1 \}$. Based on this equation and (\ref{eq5}), we have
\begin{align*}
    \zeta_i &= \frac{\partial}{\partial t} (\hat \zeta_{i-1} + R_{i-1}) + \alpha_i (\hat \zeta_{i-1} + R_{i-1})\\
    &=  \frac{\partial}{\partial t} \hat \zeta_{i-1} + \frac{\partial}{\partial t} R_{i-1} + \alpha_i (\hat \zeta_{i-1} + R_{i-1})\\
    % &=  \frac{\partial}{\partial t} \hat \zeta_{i-1} + \alpha_i(\hat \zeta_{i-1}) + \frac{\partial}{\partial t} R_{i-1} + \alpha_i (\hat \zeta_{i-1} + R_{i-1}) - \alpha_i(\hat \zeta_{i-1})\\
    &= \hat \zeta_i + \underbrace{\frac{\partial}{\partial t} R_{i-1} + \alpha_i (\hat \zeta_{i-1} + R_{i-1})-\alpha_i(\hat \zeta_{i-1})}_{= R_i}.
\end{align*}
Note that for the time derivative of residuals in (\ref{eq11}), we have the following by the definition of high-order Lie derivatives
    \begin{align*}
        \frac{\partial}{\partial t} \Delta_i = \frac{\partial}{\partial t} L_f^i h - \frac{\partial}{\partial t} L_{\hat f}^i h = L_f^{i+1} h - L_{\hat f}^{i+1} h = \Delta_{i+1},
    \end{align*}
for $i \in \{1, \dots, r-2 \}$. Thus, taking time derivative of $R_{i-1}$ and multiplying by $k_i$, we have $R_i = \frac{\partial}{\partial t}R_{i-1} + k_i R_{i-1}$. Substituting this equation into the resulting expression for $R_i$, yields
\begin{equation*}
   \alpha_i (\hat \zeta_{i-1} + R_{i-1}) - \alpha_i(\hat \zeta_{i-1}) = k_i R_{i-1},
\end{equation*}
where we are looking for functions $\alpha_i$. This functional equation can be solved if the class $\mathcal{K}$ functions $\alpha_i$ satisfy Cauchy's additive functional equation
\begin{equation}
    \alpha_i(x+y) = \alpha_i(x) + \alpha_i(y),
    \label{eq15}
\end{equation}
for all $x,y$ in their domain. Next, we will show that (\ref{eq15}) only holds for $\alpha_i = k_i x$.
Since $\alpha_i$ satisfy (\ref{eq15}), by induction and simple calculation we can show that $\alpha_i(qx) = q\alpha_i(x)$, for all $q \in \mathbb{Q}, q>0$, where $\mathbb{Q}$ is the set of all rational numbers. Since $\alpha_i$ is strictly increasing the following limit exists
\begin{equation*}
    \lim_{x \rightarrow 0^{+}} \alpha_i(x) = \lim_{n \rightarrow +\infty} \alpha_i(\frac{1}{n}) = \lim_{n \rightarrow +\infty} \frac{1}{n}\alpha_i(1) = 0. 
\end{equation*}
Then, pick any $x$ from the domain and let $\{q_n\} \in \mathbb{Q}$ be any decreasing sequence converging to $x$. From the above results, we have that $\alpha_i(q_n) = q_n\alpha_i(1)= x\alpha_i(1)$. On the other hand, from (\ref{eq15}), we can write $\alpha_i(q_n) = \alpha_i(x) + \alpha_i(q_n - x)$. We know that $q_n - x \rightarrow 0^+$ as $n \rightarrow +\infty$, thus $\alpha_i(q_n) = \alpha_i(x)$. Thus, we can conclude $\alpha_i(x) = \alpha_i(1)x = k_ix$.

For the reverse direction, suppose that we have $\alpha_i(x)=k_ix$, for $i \in \{1, \dots, r \}$. We use mathematical induction to prove that $\zeta_i = \hat{\zeta}_i + R_i\left(\Delta_i,\dots, \Delta_1\right)$, where $R_i$ is given by (\ref{eq13}). We start from $i=1$:
   $$ \zeta_1 = \dot \zeta_0 + \alpha_1(\zeta_0) = \dot h + k_1 h = \hat{\dot h} + \Delta_1 + k_1 h
    = \hat \zeta_1 + \Delta_1.$$
Thus, $R_1(\Delta_1) = \Delta_1$. Now, assume that $\zeta_p = \hat{\zeta}_p + R_p\left(\Delta_p,\dots, \Delta_1\right)$ holds for $i = p$, where $R_p\left(\Delta_p,\dots, \Delta_1\right)$ is given by (\ref{eq13}). Then for $i = p+1$, from (\ref{eq5}) we have
\begin{align*}
    \zeta_{p+1} &= \dot \zeta_p + \alpha_{p+1}\left(\zeta_p \right ) = \frac{\partial}{\partial t}\left ( \hat \zeta_p + R_p \right ) + k_{p+1} \left ( \hat\zeta_p + R_p \right )\\
    &= \hat \zeta_{p+1} + \frac{\partial}{\partial t} R_p + k_{p+1} R_p.
\end{align*}
Using the obtained relationship for the time derivative of residuals, we have
\begin{align}
   \zeta_{p+1} &= \hat \zeta_{p+1} + \frac{\partial}{\partial t} R_p + k_{p+1} R_p\nonumber\\
   &= \hat \zeta_{p+1} + \Delta_{p+1} + \sum_{1\leq j_1 \leq p}k_{j_1}\Delta_{p}+\dots+\nonumber \\
   &\sum_{1\leq j_1<\dots<j_{p-1}\leq p}{k_{j_1}k_{j_2}\dots k_{j_{p-1}}} \Delta_2 \nonumber\\
   +& k_{p+1} (\Delta_p + \sum_{1\leq j_1 \leq p}k_{j_1}\Delta_{p-1}+\dots+ \nonumber\\
   &\sum_{1\leq j_1<\dots<j_{p-1}\leq p}{k_{j_1}k_{j_2}\dots k_{j_{p-1}}} \Delta_1).\label{eq14}
\end{align}
By factorization, we have
\begin{align*}
 \zeta_{p+1} &=  \hat \zeta_{p+1} + \Delta_{p+1} + \sum_{1\leq j_1 \leq p+1}k_{j_1}\Delta_{p}+\dots+ \\
   &\sum_{1\leq j_1<\dots<j_p\leq p}{k_{j_1}k_{j_2}\dots k_{j_p}} \Delta_1, 
\end{align*}
which proves our claim for $i = p+1$.
\end{proof}
\section{PROPOSED GP-BASED HOCBF DESIGN}
Theorem \ref{thm1} shows that we can approximate the effect of uncertainty on the high-order safety certificate by supervised learning only if we choose linear class $\mathcal{K}$ functions in our HOCBF design. It also sets the basis for reducing the complexity of our proposed learning method.

By substituting the results from Theorem \ref{thm1} in ($\ref{eq5})$ recursively and using (\ref{eq11}) and (\ref{eq12}), we can rewrite $\zeta_r(\bm x, \bm u)$ as
\begin{align*}
    \zeta_r(\bm x, \bm u) &= \hat \zeta_r(\bm x, \bm u) + \underbrace{R_r(\Delta_r, \dots, \Delta_1) + \Delta_g(\bm x) \bm u}_{=\Delta_{hocbf}(\bm x, \bm u)},
    % \label{eq16}
\end{align*}
where $\Delta_i:\mathbb{R}^n \rightarrow \mathbb{R}, i\in \{1, \dots, r\}$ and $\Delta_g:\mathbb{R}^n \rightarrow \mathbb{R}^{1 \times m}$, which gives us the high-order safety certificate of the form
\begin{equation}
    \hat \zeta_r(\bm x, \bm u) + \Delta_{hocbf}(\bm x, \bm u)\geq 0.
    \label{eq16}
\end{equation}
The result in (\ref{eq16}) enables us to compensate for the adverse effect of the uncertainty in the safety of the true system by learning the residual term $\Delta_{hocbf}(\bm x, \bm u)$.
\begin{rem}\label{rem1}
    In many data-driven safety-critical control methods, the target functions are the unknown dynamics $f, g$ in (\ref{eq1}). However, $\Delta_{hocbf}$ is a scalar value. This fact substantially reduces the complexity of the learning process by addressing a lower-dimensional problem compared to the latter approach.
\end{rem}
Substituting $i = r$ in (\ref{eq13}) and expanding $\Delta_{hocbf}(\bm x, \bm u)$ yields
\begin{equation}
    \Delta_{hocbf}(\bm x, \bm u) = \gamma_r \Delta_r + \gamma_{r-1} \Delta_{r-1} + \dots + \gamma_1 \Delta_1+ \Delta_g(\bm x) \bm u,
    \label{eq17}
\end{equation}
where we denote coefficients of $\Delta_i$'s in (\ref{eq13}) by $\gamma_i, i \in \{1, \dots, r\}, \gamma_r = 1$, to simplify notations.

We can approximately measure $\Delta_i$ using finite difference methods and collect a dataset for our learning method:
\begin{align}
    \tilde{\Delta}_j &= FD^j(\bm x(t), \Delta t) - \hat{h}^j(\bm{x}(t)), j \in \{1, \dots, r-1 \},\nonumber\\
    \tilde{\Delta}_r &= FD^r(\bm x(t), \bm{u}(t), \Delta t) - \hat{h}^r(\bm{x}(t), \bm{u}(t)),
    \label{eq18}
\end{align}
where $\Delta t$ is the sampling interval.
Then, the approximated measurement can be derived by substituting (\ref{eq18}) into (\ref{eq17}). 
\subsection{GP for real-valued functions}\label{3subsec1}
GPs provide a framework to approximate nonlinear functions. We represent a GP approximation of $v:\mathcal{X} \rightarrow \mathbb{R}$ as $v(\bm x) \sim \mathcal{GP}\left(m(\bm x), k(\bm x, \bm x') \right)$, which is fully specified by its mean function $m:\mathcal{X} \rightarrow \mathbb{R}$ and covariance (kernel) function $k:\mathcal{X} \times \mathcal{X} \rightarrow \mathbb{R}$. Our prior knowledge of the problem can shape the form of $m(\bm x)$ and $k(\bm x, \bm x')$, considering that the covariance function must satisfy positive definiteness \cite{williams2006gaussian}. Given noisy measurements $w_j = v(\bm x_j) +  \varepsilon_j, j\in \{1, \dots, N \}$, which are corrupted by Gaussian noise $\varepsilon_j \sim \mathcal{N}(0,\,\sigma^{2})$, a GP model can infer a posterior mean and variance for a test point $\bm x_t$ conditioned on the measurements
\begin{align} 
    \mu(\bm x_t) &= {{\bm{w}}^T}{\left( {K + \sigma _n^2 I} \right)^{-1}}\bar {K}^T,\nonumber \\ 
    \sigma(\bm x_t)^2 &= k\left( {{x_t},{x_t}} \right) - {\bar K}{\left( {K + \sigma _n^2 I} \right)^{- 1}}\bar{K}^T,
    \label{eq19}
\end{align}
where $\bm w \in \mathbb{R}^N$ is the vector of measurements $w_j$, $K\in \mathbb{R}^{N \times N}$ is the Gram matrix with elements $K_{ij} = k(\bm x_i, \bm x_j)$, and $\bar K = \begin{bmatrix}
    k(x_t, x_1), \dots, k(x_t, x_N)
\end{bmatrix}^T \in \mathbb{R}^N$. We considered zero prior mean to obtain (\ref{eq19}).
\subsection{GP for uncertainty-aware HOCBF}
Going back to the original problem, our objective is to approximate $\Delta_{hocbf}$ using GP. From (\ref{eq17}), we can verify that $\Delta_{hocbf}$ is control-affine and can be characterized by
\begin{equation*}
    \Delta_{hocbf} = \begin{bmatrix}
        \Delta_1, \dots, \Delta_r, \Delta_g
    \end{bmatrix} \begin{bmatrix}
        \bm \gamma\\ \bm u
    \end{bmatrix} \coloneqq \begin{bmatrix}
        \Delta_1, \dots, \Delta_r, \Delta_g
    \end{bmatrix} \bm y,
\end{equation*}
where $\bm \gamma = \begin{bmatrix}
    \gamma_1& \dots& \gamma_r
\end{bmatrix}^T\in \mathbb{R}^r$, also we denote the concatenation of $\bm \gamma, \bm u$ by $\bm y\in \mathbb{R}^{m+r}$. This prior knowledge of the problem can be embedded into the structure of our kernel function, which facilitates a more accurate regression.
\begin{prop}\label{prop2}
     Let $\bm x \in \mathbb{R}^n$, $\bm y \in \mathbb{R}^{m+r}$, and define the input domain $\mathcal{X} = \mathbb{R}^{n+m+r}$. Then consider the real-valued function $k_c: \mathcal{X} \times \mathcal{X} \rightarrow \mathbb{R}$, defined by
    \begin{equation}
        k_c \left(\begin{bmatrix}
        \bm x\\\bm y
    \end{bmatrix}, \begin{bmatrix}
        \bm x'\\ \bm y'
    \end{bmatrix} \right ) = \bm y^T \Lambda(\bm x, \bm x') \bm y', 
    \label{eq20}
    \end{equation}
    where $\Lambda(\bm x, \bm x') = \diag(\begin{bmatrix}
        k_1(\bm x,\bm x'), \dots, k_{m+r}(\bm x,\bm x')\end{bmatrix} )$ and $k_i: \mathbb{R}^n \times \mathbb{R}^n \rightarrow \mathbb{R}, i\in\{1, \dots, m+r \}$. If $k_i$'s are positive-definite kernels, then $k_c$ is a positive-definite kernel.
\end{prop}

\begin{proof}
    Since $k_i(\bm x, \bm x')$'s are positive definite (valid) kernels, then by the definition, there is a feature map $\varphi(\bm x)$ such that $k_i(\bm x, \bm x') = \varphi^T_i(\bm x) \varphi_i(\bm x')$ for $i\in \{1, \dots, m+r \}$. Now, Let $y^i$ and $y'^i$ be the $i^{th}$ element of the corresponding vectors. We have
    \begin{align*}
        k_c &= \bm y^T \diag \left([
        \varphi^T_1(\bm x) \varphi_1(\bm x'), \dots, \varphi^T_{m+r}(\bm x) \varphi_{m+r}(\bm x')
    ] \right) \bm y'\\
    &= \sum_{i=1}^{m+r} y^i \varphi_i^T(\bm x) \varphi_i(\bm x') y'^i\\
    &= \begin{bmatrix}
         y^1 \varphi_1^T(\bm x) & \dots & y^{m+r} \varphi_{m+r}^T(\bm x)
    \end{bmatrix}
    \begin{bmatrix}
        y'^1 \varphi_1(\bm x') \\ \vdots \\ y'^{m+r} \varphi_{m+r}(\bm x')
    \end{bmatrix}\\
        &= \psi^T\left(\begin{bmatrix}
        \bm x\\\bm y
    \end{bmatrix} \right ) 
    \psi\left(\begin{bmatrix}
        \bm x'\\ \bm y'
    \end{bmatrix} \right ),
    \end{align*}
    which again by the definition of kernels, proves that $k_c(\cdot, \cdot)$ is a positive-definite (valid) kernel.
\end{proof}
An immediate consequence of Proposition \ref{prop2} is that $k_c$ is the reproducing kernel of a reproducing kernel Hilbert space (RKHS) $\mathcal{H}_{k_c}(\mathcal{X})$.

A dataset can be generated by collecting trajectories of the true system. We denote it by $\mathcal{D} = \{ ((\bm x_j, \bm y_j), z_j) \}_{j=1}^N$, where $(\bm x_j, \bm y_j) \in \mathcal{X}\subset \mathbb{R}^n \times \mathbb{R}^{m+r}$ is the input data and outputs $z_j$ are obtained from (\ref{eq17}) and (\ref{eq18}). Also, we define $X = \begin{bmatrix} \bm x_1, \dots, \bm x_N\end{bmatrix}\in \mathbb{R}^{n \times N}$, and $Y = \begin{bmatrix} \bm y_1, \dots, \bm y_N\end{bmatrix}\in \mathbb{R}^{(m+r) \times N}$. Based on the collected dataset, the GP model which is equipped with the kernels of the form $k_c$ in Proposition \ref{prop2}, gives the following expression for the posterior distribution of a test point $(\bm x_*, \bm y_*)$:
\begin{align}
    m(\bm x_*, \bm y_*) &= \bm z^T \left(K_c + \sigma_n^2 I \right)^{-1}\bar{K}^T\bm y_*,\label{eq21}\\
    \sigma^2(\bm x_*, \bm y_*) &= \bm y_*^T \left ( \Lambda(\bm x_*,\bm x_*) - \bar{K}\left(K_c + \sigma_n^2 I \right)^{-1}\bar{K}^T \right ) \bm y_*,\label{eq22}
\end{align}
where $K_c$ is the Gram matrix of $k_c(\cdot, \cdot)$ for the input data pair $(X,Y)$, and $\bar K \in \mathbb{R}^{(m+r)\times N}$ is given by
\begin{align*}
    \bar K &= \begin{bmatrix}
        \bm{\bar{k}_1} & \bm{\bar{k}_2} & \dots &\bm{\bar{k}}_N
    \end{bmatrix} 
    \circ Y,\\
    \bm{\bar{k}_i} &= \begin{bmatrix}
        k_1(x_*, x_i), \dots, k_{m+r}(x_*, x_i)
    \end{bmatrix}^T, i \in \{1, \dots, N \},
\end{align*}
where $\circ$ denotes the element-wise multiplication (Hadamard product) of two matrices with identical dimensions.
This result highlights a key advantage of using $k_c$ kernel. The resulting expression for the posterior mean $m(\bm x_*, \bm y_*) = \bm \mu(\bm x_*)\bm y_*$ and the variance $\sigma^2(\bm x_*, \bm y_*) =\bm y_*^T \Sigma(\bm x_*) \bm y_*$ are linear and quadratic in $\bm y_*$ (and the control input), respectively. We will exploit this feature in the next section to establish a convex GP-based safety filter based on HOCBFs.
\subsection{Confidence bounds for the estimation of the effect of uncertainty}

Although GP is inherently a probabilistic model, a high probability error bound can be derived for the distance between the true value and the GP prediction. This requires an additional assumption on $\Delta_{hocbf}$.

\begin{assum}\label{assum2}
    Each $\Delta_i, i \in \{1, \dots, r \}$ and each $i^{th}$ element of $\Delta_g$ for $i \in \{1, \dots, m \}$ is a member of $\mathcal{H}_{k_i}$ with bounded RKHS norm $\| \Delta_i \|_{k_i} \leq \eta$ for $i \in \{1, \dots, m+r \}$.
\end{assum}

\begin{lem}[\cite{srinivas2009gaussian}]
 Let Assumption \ref{assum2} hold. Then, with a probability of at least $1-\delta$, the following holds
    % Then, the GP approximation error is bounded as
        \begin{equation} 
            \vert m(\bm x, \bm y)-\Delta_{hocbf}(\bm x, \bm u)\vert \leq\beta\sigma(\bm x, \bm y), 
            \label{eq23}
        \end{equation}
    on a compact set ${\mathcal{D}} \subset \mathcal{X}, \delta \in (0,1)$ and $\beta=\sqrt{2\eta^2+300 \kappa_{N+1} \log ^3((N+1) / \delta)}$, where $\kappa_{N+1}$ is the maximum mutual information that can be obtained after getting $N+1$ data, and $\eta$ is the upper bound of the corresponding RKHS norm, and $m(\bm x, \bm y)$ and $\sigma(\bm x, \bm y)$ are the posterior mean and standard deviation of a test point $(\bm x, \bm y) \in \mathcal{X}$.
\end{lem}
Based on the probabilistic bounds on (\ref{eq23}) and the high-order safety certificate (\ref{eq16}), we can conclude that the following hold with a probability of at least $1 - \delta$
\begin{equation}
    \zeta_r(\bm x, \bm u) \geq \hat{\zeta}_r(\bm x, \bm u) + m(\bm x, \bm y) - \beta \sigma(\bm x, \bm y),
    \label{eq24}
\end{equation}
where $m(\bm x, \bm y)$ and $\sigma(\bm x, \bm y)$ are the mean and standard deviation of a query point derived from (\ref{eq21}), (\ref{eq22}). Incorporating (\ref{eq24}) into the optimization problem (\ref{eq9}), we have
\begin{align}
    \boldsymbol{u}_{s} ={} &\underset{\boldsymbol{u} \in \mathcal{R}^m}{\arg\min} \hspace{4pt} {\| \bm u - \bm u_{nom} \|^2_2} \nonumber\\
    &\textrm{s.t.} \quad \hat{\zeta}_r(\bm x, \bm u) + m(\bm x, \bm y) - \beta \sigma(\bm x, \bm y) \geq 0.
    % &\textrm{s.t.} \quad L_f^r h+ L_gL_f^{r- 1} h \hspace{2pt} \bm u + \sum _{i = 1}^{r - 1}L_f^i(\alpha _{r - i}\circ \zeta _{r-i-1})\nonumber\\ &\quad \quad+\alpha_r(\zeta_{r-1}) \geq 0.
    \label{eq25}
\end{align}
Note that the constraint in (\ref{eq25}) is constructed regardless of the underlying true dynamics.
\begin{thm}\label{thm2}
    The optimization problem (\ref{eq25}) is convex and can be converted into the standard second-order cone program (SOCP) of the form (\ref{eq26}), if $m(\bm x, \bm y)$ and $\sigma(\bm x, \bm y)$ satisfy (\ref{eq21}) and (\ref{eq22}).
    \begin{align}
       \bm u_{socp} ={} &\underset{\bm \omega}{\arg\min} \hspace{4pt} \bm f^T \bm \omega \nonumber\\
        \text { s.t. } 
        & \left\|M_1 \bm \omega+\bm{n}_1\right\|_2 \leq \bm{p}_1^T \bm \omega,\nonumber\\
        & \left\|M_2 \bm \omega+\bm{n}_2\right\|_2 \leq \bm{p}_2^T \bm \omega+q_2,
        % \quad i=1, \ldots, N_c,
        \label{eq26}
    \end{align}
    where $\bm \omega \in \mathbb{R}^{m+1}$ and $\bm f \in \mathbb{R}^{m+1}$, and the matrix $M_i$ and $\bm{n}_i, \bm{p}_i, q_i$, for $i=1,2$, have appropriate dimensions.
\end{thm}
\begin{proof}
    % We need to show that (\ref{eq25}) can be transformed 
    Let's denote the objective function of (\ref{eq25}) by $J=\left\|\bm u-\bm u_{nom}\right\|_2^2$. Since the Euclidean norm is always positive, we can consider the equivalent objective function $J_1=\left\|\bm u-\bm u_{nom}\right\|_2$. We set $\left\|\bm u-\bm u_{nom}\right\|_2 \leq t$, where $t \in \mathbb{R}$ is an auxiliary variable and convert it to a second-order cone constraint. Now, we need to solve a new minimization problem with new augmented variable $\bm \omega = \begin{bmatrix}\bm u^T & t\end{bmatrix}^T\in \mathbb{R}^{m+1}$ as
    \begin{align}
        \min & \begin{bmatrix}\bm 0_m^T &1 \end{bmatrix} \bm \omega \nonumber\\
        \text { s.t. } & \|\underbrace{\begin{bmatrix}I_m & \bm 0_m\end{bmatrix}}_{\coloneqq M_1} \bm \omega\underbrace{-\bm{u}_{nom}}_{\coloneqq\bm n_1}\| \leqslant \underbrace{\begin{bmatrix} \bm 0_m^T & 1\end{bmatrix}}_{\coloneqq \bm p_1^T} \bm \omega
        \label{eq27},
    \end{align}
    where $I_m$ is a $m \times m$ identity matrix and $\bm 0_m \in \mathbb{R}^m$ is a vector of zeros, and $\bm f =  \begin{bmatrix}\bm 0_m^T &1 \end{bmatrix}^T$.

    Next, we need to show that the constraint in (\ref{eq25}) is a SOC constraint. Note that based on (\ref{eq21}), we have $m(\bm x, \bm y) = \bm \mu (\bm x) \bm y$, which can be written as
    \begin{align*}
        m(\bm x, \bm y) = \bm \mu_{r} \bm \gamma + \bm \mu_{m} \bm u,
    \end{align*}
    where the notation $\bm \mu_{r}$ refers the first $r$ elements, and $\bm \mu_{m}$ refers the last $m$ elements of the row vector $\bm \mu \in \mathbb{R}^{1 \times (m+r)}$. Also, from (\ref{eq8}), we know that $\hat \zeta_r$ is control affine, i.e. $\hat{\zeta}_r = \bm{\hat{\zeta}}_f \bm \gamma + \bm{\hat{\zeta}}_g \bm u$, where the row vectors $\bm{\hat{\zeta}}_f \in \mathbb{R}^{1\times r}$ and $\bm{\hat{\zeta}}_g \in \mathbb{R}^{1 \times m}$. Thus, the right hand side of the inequality in (\ref{eq25}) is affine in $\bm u$ as desired.

    Based on (\ref{eq22}), we have that $\sigma^2(\bm x, \bm y)= \bm y^T \Sigma(\bm x) \bm y$, where $\Sigma \in \mathbb{R}^{(m+r) \times (m+r)}$. Since $k_c$ is a valid kernel and $\sigma_n>0$, $\Sigma$ is positive definite. Then, we have
    \begin{equation*}
         \sigma(\bm x, \bm y)=\sqrt{\bm y^T \Sigma \bm y}=\sqrt{\bm y^T L^T L \bm y}=\|L \bm y\|_2,
    \end{equation*}
    where $L \in \mathbb{R}^{(m+r) \times (m+r)}$ is the matrix square root of $\Sigma$. By definition $\bm y=\begin{bmatrix} \bm \gamma^T & \bm u^T \end{bmatrix}^T$, thus we have $\sigma(\bm x, \bm y) = \left\|L^r \bm \gamma+L^m \bm u\right\|_2$, where the notation $L^r$ and $L^m$ refers to the first $r$ columns, and the last $m$ columns of the matrix $L$, respectively. Now, we can rewrite the high-order safety certificate (\ref{eq25}) as a SOC constraint of the form
    \begin{equation}
        \left\|A(\bm x) \bm u+ \bm{b}(\bm x)\right\|_2 \leq \bm{c}(\bm x) \bm u+d(\bm x), 
        \label{eq28}
    \end{equation}
    where 
    \begin{align*}
        A(\bm x) &= \beta L^m \in \mathbb{R}^{(m+r) \times m},\\
        \bm b(\bm x) &= \beta L^r \bm \gamma \in \mathbb{R}^{m+r},\\
        \bm c(\bm x) &= \bm{\hat \zeta}_g + \bm \mu_m \in \mathbb{R}^{1 \times m},\\
        d(\bm x) &= (\bm{\hat \zeta}_f + \bm{\mu}_r) \bm \gamma \in \mathbb{R}.
    \end{align*}
    It can be expressed in terms of the new variable $\bm \omega$ as
    \begin{equation*}
        \| \underbrace{\begin{bmatrix}
        A(\bm x) & \bm{0}_{m+r}
        \end{bmatrix}}_{\coloneqq M_2} \bm \omega+\underbrace{\bm{b}(\bm x)}_{\coloneqq \bm n_2}\|_2 \leq \underbrace{\begin{bmatrix}\bm{c}(\bm x) & 0 \end{bmatrix}}_{\coloneqq \bm p_2^T} \bm \omega+ \underbrace{d(\bm x)}_{\coloneqq q_2}.
    \end{equation*}
Adding the resulting constraint to (\ref{eq27}) leads to (\ref{eq26}), which concludes the proof.
\end{proof}
\subsection{Feasibility analysis}
In this section, we will analyze the point-wise feasibility of the SOCP (\ref{eq25}). Intuitively, if the GP prediction is less accurate, the resulting standard deviation $\sigma(\bm x, \bm y)$ will be large, leading to a more conservative approach. In particular, it restricts the space on which a safe control input can be selected. In consequence, it may make the SOCP infeasible. In the following, we theoretically analyze the conditions for feasibility.
\begin{thm}\label{thm3}
    Given a point $\bm x \in \mathbb{R}^n$ and a HOCBF design with the coefficient vector $\bm \gamma \in \mathbb{R}^r$, the SOCP (\ref{eq25}) is feasible, if and only if there exists a control input $\bm u \in \mathbb{R}^m$ such that it satisfies the following conditions
        \begin{align}
        % $$\begin{equation}
               & \begin{bmatrix}
                \bm{\hat{\zeta}}_f + \bm \mu_r & \bm{c}(\bm x)
                \end{bmatrix}
                \begin{bmatrix}
                    \bm \gamma\\ \bm u
                \end{bmatrix}
                \geq 0,\label{eq29}\\
            &\begin{bmatrix}
                \bm \gamma^T & \bm u^T
            \end{bmatrix} S(\bm x) \begin{bmatrix}
                \bm \gamma \\ \bm u
            \end{bmatrix} \leq 0,
            \label{eq30}
        \end{align}
        where $S \in \mathbb{R}^{(m+r) \times (m+r)}$ is of the form
        \begin{align}
            S(\bm x) &= \begin{bmatrix}
            S_1(\bm x) & S_2(\bm x) \\
            S_2^T(\bm x) & S_3(\bm x)
            \end{bmatrix},\nonumber\\
            S_1(\bm x) &= \beta^2 {L^r}^T L^r-(\bm{\hat{\zeta}}_f+\bm \mu_r)^T(\bm{\hat{\zeta}}_f+\bm \mu_r),\nonumber\\
            S_2(\bm x) &= \beta {L^r}^T A(\bm x)-(\bm{\hat{\zeta}}_f+\bm \mu_r)^T \bm c(\bm x),\nonumber\\
            S_3(\bm x) &= A(\bm x)^T A(\bm x)-\bm c(\bm x)^T \bm c(\bm x).
            \label{eq31}
        \end{align}
\end{thm}

\begin{proof}
    In order to analyze the feasibility of the SOCP (\ref{eq25}), we need to verify that the safety SOC constraint (\ref{eq28}) is feasible. Since the left hand side is always non-negative, the right hand side must be also non-negative, which leads to the first condition.

    Now, since both sides of (\ref{eq28}) are non-negative, we can square both sides of the equation and then gather all the terms on the left-hand side, Factorizing the similar terms gives the second condition.
\end{proof}
Based on this result, we obtain the necessary condition for point-wise feasibility as stated in the following:
\begin{lem}\label{lem2}
    Given $\bm x \in \mathbb{R}^n$, if the SOCP (\ref{eq25}) is feasible, then the following condition must be satisfied
    \begin{equation}
        1 - \frac{1}{\beta^2} \bm{\phi} \Sigma^{-1} \bm{\phi}^T \leq 0,
    \end{equation}
    where $\bm{\phi} = \begin{bmatrix}
    \bm{\hat{\zeta}}_f + \bm{\mu}_r & \bm{c}(\bm{x})\end{bmatrix} \in \mathbb{R}^{1 \times (m+r)}$.
\end{lem}
\begin{proof}
We prove this by contradiction. Let's assume that there exists a solution $\bm u \in \mathbb{R}^m$ of the SOCP (\ref{eq25}), which satisfies $1 - \frac{1}{\beta^2} \bm{\phi} \Sigma^{-1} \bm{\phi}^T > 0$. 
Let's define the matrix $M$ partitioned as
\begin{equation*}
    M = \begin{bmatrix}
        1 & \bm{\phi} \\
        \bm{\phi}^T & \beta^2 \Sigma(\bm{x})
    \end{bmatrix} \in \mathbb{R}^{(m+r+1) \times (m+r+1)}.
\end{equation*}
The Schur complements of $M$ is obtained by
\begin{equation*}
    \left\{
    \begin{array}{l}
        M / 1 = \beta^2 \Sigma(\bm{x}) - \bm{\phi}^T \bm{\phi}, \\
        M / \beta^2 \Sigma(\bm{x}) = 1 - \frac{1}{\beta^2} \bm{\phi} \Sigma^{-1} \bm{\phi}^T.
    \end{array}
    \right.
\end{equation*}
We can easily verify that $\beta^2 \Sigma=\begin{bmatrix}\beta L^r & A(\bm x)\end{bmatrix}^T\begin{bmatrix}\beta L^r & A(\bm x)\end{bmatrix}$ and $M / 1=S(\bm x)$.
We know that $\Sigma(\bm x)$ is symmetric positive-definite (SPD). So, $M$ is also symmetric and we can characterize its definiteness by the Schur complement theorem. 
Since its lower right block is positive definite and its corresponding Schur complement $M/\beta^2 \Sigma(\bm x)>0$ by our assumption, we can conclude that $M$ is also positive definite.

Now, given that $M$ is SPD, and its upper left block is always positive definite, we can conclude that $M/1 = S(\bm x)$ must be positive definite, which is a contradiction to the condition (\ref{eq30}). Thus, $1 - \frac{1}{\beta^2} \bm{\phi} \Sigma^{-1} \bm{\phi}^T \leq 0$ holds.
\end{proof}
Theorem $\ref{thm3}$ also provides the theoretical background to obtain a sufficient condition of feasibility.
\begin{lem}\label{lem4}
    Given $\bm x \in \mathbb{R}^n$, the SOCP (\ref{eq25}) is feasible at $\bm x$, if $S_3(\bm x)$ is negative definite.
\end{lem}
\begin{proof}
    From (\ref{eq31}), we know that $S_3(\bm x)$ is a symmetric matrix, thus it has real eigenvalues. Let's denote its maximum eigenvalue and the corresponding eigenvector by $\lambda_{m}$, $\bm{e}_{m}$. Then, by the assumption that $S_3(\bm x)$ is negative definite, we can conclude that $\lambda_{m}<0$ and $\bm{e}_{m}^T S_3(\bm x) \bm{e}_{m}<0$. Substituting (\ref{eq31}) into the former, we have
    \begin{equation*}
        \bm{e}_{m}^T A^T(\bm x) A(\bm x) \bm{e}_{m} - (\bm c(\bm x) \bm{e}_{m})^T(\bm c(\bm x)\bm{e}_{m})<0.
    \end{equation*}
    Since $A(\bm x)^TA(\bm x)$ is positive definite, $\bm c(\bm x) \bm{e}_{m} \neq 0$ must hold. We take a control input in the direction of this eigenvector as $\bm u_e = \alpha \sgn(\bm{c}(\bm x) \bm{e}_{m})\cdot\bm{e}_{m}, \alpha>0$. Next, we need to check if the resulting control input can satisfy the conditions of Theorem \ref{thm3}. Substituting $\bm{u}_e$ in (\ref{eq30}), we have
    \begin{equation*}
        \bm \gamma^T S_1 \bm \gamma + 2\alpha \gamma^T S_2 ((\bm{c}(\bm x) \bm{e}_{m})\cdot\bm{e}_{m}) + \alpha^2 \bm e_{m}^T S_3(\bm x) \bm e_{m}.
    \end{equation*}
    By choosing large enough $\alpha$, the above statement can be made negative since $\bm e_m^T S_3(\bm x) \bm e_m <0$.
    Also, substituting $\bm u_e$ into (\ref{eq29}), we have
    \begin{equation*}
        \bm c(\bm x)  \bm u_e + (\bm{\hat{\zeta}}_f + \bm \mu_r) \bm \gamma = \alpha \lvert \bm c(\bm x) \bm e_m \rvert + (\bm{\hat{\zeta}}_f + \bm \mu_r) \bm \gamma.
    \end{equation*}
    Again, by choosing a large enough $\alpha$, we can make the above expression positive, which concludes the proof.
\end{proof}

\begin{figure}[t]
    \centering
    \includegraphics[width=8.5cm]{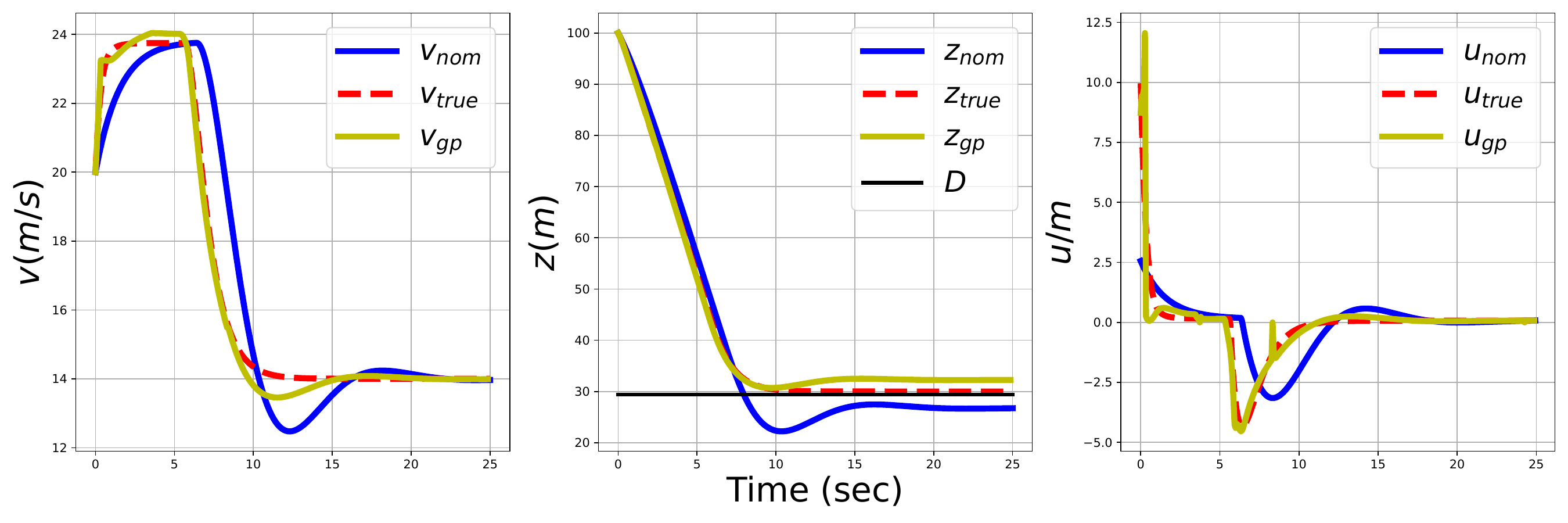}
    \caption{The state variables and control input in the presence of uncertainty when applying the GP-based SOCP-HOCBF (yellow) compared to the nominal QP-HOCBF (blue) and the QP-HOCBF design based on the true system (red dashed).}
    \label{fig1}
\end{figure}
\section{Simulation results}
In this section, two simulation studies highlight the effectiveness of the proposed method. Simulations are run in the Python environment, and the optimization problems are solved using the CVXPY \cite{Diamond2016} package.

\subsection{Adaptive cruise control with collision avoidance}
Consider an adaptive cruise control (ACC) modeled by
\begin{equation}
    \dot{\bm{x}} = \begin{bmatrix}
        -\frac{f_0 + f_1 v + f_2 v^2}{m}\\ v_0 - v
    \end{bmatrix} + \begin{bmatrix}
        \frac{1}{m} \\ 0
    \end{bmatrix} u,
    \label{eq33}
\end{equation}
where the system state $\bm x = \begin{bmatrix} v & z\end{bmatrix}^T \in \mathbb{R}^2$  containing the ego car's forward velocity $v$, and the distance between the ego car and the preceding car, denoted by $z$. The wheel force $u \in \mathbb{R}$ is the control input of the ego vehicle. The model assumes a constant velocity $v_0$ for the preceding car; however, there exists uncertainty in the estimation of $v_0$. The mass of the ego car is denoted by $m$ and it experiences rolling resistance modeled as $f_0+f_1 v+ f_2 v^2$.

The control objective is to reach a target speed $v_d$ while maintaining a safe distance $D$ to the vehicle in front. We are provided with a nominal model characterized with $m = 825\hspace{2pt}kg$, and $f_0 = 0.1, f_1 = 5, f_2 = 0.25$, and $v_0 = 16\hspace{2pt}m/s$. The true system parameters, which are unknown, are  $m = 3300\hspace{2pt}kg$, $f_0 = 0.2, f_1 = 10, f_2 = 0.5$, and $v_0 = 14\hspace{2pt}m/s$.

We consider the Lyapunov function $V(\bm x) = (v - v_d)^2$, and use GP-based control Lyapunov function approach \cite{castaneda2021gaussian} to stabilize the system near the desired velocity $v_d = 24\hspace{2pt}m/s$. GP-based CLF can be converted to a SOC constraint and incorporated to the $(\ref{eq25})$ as a soft constraint. We have $h(\bm x) = z - D$, where $D = 30\hspace{2pt}m$, as the CBF of the system, which has relative degree two, with respect to the system (\ref{eq33}). The vector of coefficients is $\bm \gamma = \begin{bmatrix}4 & 3.75\end{bmatrix}^T$ in the HOCBF design.
Figure \ref{fig1} illustrates the performance of the proposed method in comparison to the nominal safety-critical controller. It also includes the result of the oracle HOCBF-QP that designed by the true dynamics as a reference. 

The simulation is started from the initial point $\bm{x}_0 = \begin{bmatrix} 20, 100 \end{bmatrix}^T$. First, the controller increases the velocity to converge to the desired speed, which decreases $z$. The distance between the ego car and the preceding car reaches the limit after $7$ seconds of the simulation. Due to the adverse effect of uncertainty, the nominal safety controller violates the safe distance. However, the GP-based controller could avoid unsafe behavior. It can be verified that the proposed method could recover the true system's performance.

We used the episodic data collection method. For the first episode, we run the system using nominal QP-HOCBF design and collect the system trajectory, until it becomes unsafe. Then, we use the proposed GP model with a high probability bound ($1 - \delta = 0.95$) and apply the SOCP controller, which contains an uncertainty-aware safety constraint. We add the collected data to the dataset at each episode, until the system passes the simulation time without any unsafe interruption. A total number of $119$ samples are collected for this simulation. 
We used infinitely differentiable squared exponential kernel \cite{williams2006gaussian}, for $k_i, i=1,2, 3$ individual kernels in the structure of $k_c$.

\subsection{Active suspension system}
Consider two degrees of freedom, active suspension system of a quarter car (see Figure \ref{figasc} for an illustration), modeled in the form of $\dot{\bm x} = f(\bm x) + g(\bm x) u + g_d(\bm x) d$, as follows:
\begin{equation*}
    \dot{\bm x}= \begin{bmatrix}
        x_3 \\ x_4 \\ \frac{k_1(x_2 - x_1) + b\hspace{0.5pt}(x_4 - x_3)}{m_1}\\ \frac{k_1 (x_1-x_2) -k_2 x_2+ b\hspace{0.5pt}(x_3 - x_4)}{m_2}
    \end{bmatrix} + \begin{bmatrix}
        0 \\ 0 \\ \frac{1}{m_1} \\ -\frac{1}{m_2}
    \end{bmatrix} u + \begin{bmatrix}
        0 \\ 0 \\0 \\ \frac{k_2}{m_2}
    \end{bmatrix}d.
    % \label{eq34}
\end{equation*}

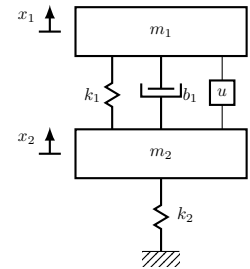
\begin{figure}[htbp]
    \centering
    \begin{tikzpicture}[every node/.style={draw,outer sep=0pt,thick}, scale=0.65, transform shape]
\tikzset{spring/.style={thick,decorate,decoration={zigzag,pre length=0.3cm,post length=0.3cm,segment length=6}},
damper/.style={thick,decoration={markings,  
  mark connection node=dmp,
  mark=at position 0.5 with 
  {
    \node (dmp) [thick,inner sep=0pt,transform shape,rotate=-90,minimum width=15pt,minimum height=3pt,draw=none] {};
    \draw [thick] ($(dmp.north east)+(2pt,0)$) -- (dmp.south east) -- (dmp.south west) -- ($(dmp.north west)+(2pt,0)$);
    \draw [thick] ($(dmp.north)+(0,-5pt)$) -- ($(dmp.north)+(0,5pt)$);
  }
}, decorate},
ground/.style={fill,pattern=north east lines,draw=none,minimum width=0.75cm,minimum height=0.3cm}}

% Bottom mass
\node (M1) [minimum width=3.5cm,minimum height=1cm] {$m_2$};

\draw[thick] (M1.west) ++(-0.3,0cm) -- +(-0.45,0cm);
\draw [-latex, thick] (M1.west) ++(-0.525,0cm) -- +(0,0.55cm);
\node (x_22) at ($(M1.west)!(M1.west)!(M1.west) + (-1cm,0.3cm)$) [draw=none] {$x_{2}$};

\node (ground2) at (M1.south) [ground,yshift=-1.5cm,anchor=north] {};
\draw (ground2.north west) -- (ground2.north east);
\draw [spring] (ground2.north) -- ($(M1.south east)!(ground2.north)!(M1.south west)$);
\node (k_2) at ($(M1.north)!(M1.north)!(M1.south) + (0.5cm,-1.75cm)$) [draw=none] {$k_{2}$};

\draw [damper] ($(M1.north)!(M1.north)!(M1.south) + (0,0cm)$) -- ($(M1.north)!(M1.north)!(M1.south) + (0,1.5cm)$);
\draw [spring] ($(M1.north)!(M1.north)!(M1.south) + (-1,0cm)$) -- ($(M1.north)!(M1.north)!(M1.south) + (-1,1.5cm)$);

\node (b_1) at ($(M1.north)!(M1.north)!(M1.south) + (0.6cm,0.7cm)$) [draw=none] {$b_{1}$};
\node (k_1) at ($(M1.north)!(M1.north)!(M1.south) + (-1.4cm,0.7cm)$) [draw=none] {$k_{1}$};

% Top mass
\node (M2) at ($(M1.north)!(M1.north)!(M1.south) + (0,2cm)$) [minimum width=3.5cm,minimum height=1cm] {$m_1$};

\node (U) at ($(M1.north)!(M1.north)!(M1.south) + (1.25,0.75cm)$) [minimum width=.5cm,minimum height=0.5cm] {$u$};
\draw ($(M1.north) +(1.25,0cm)$) -- (U.south);
\draw ($(U.north) +(0,0cm)$) -- ($(M2.south) +(1.25,0cm)$);

\draw[thick] (M2.west) ++(-0.3,0cm) -- +(-0.45,0cm);
\draw [-latex, thick] (M2.west) ++(-0.525,0cm) -- +(0,0.55cm);
\node (x_22) at ($(M2.west)!(M2.west)!(M2.west) + (-1cm,0.3cm)$) [draw=none] {$x_{1}$};

% \draw [-latex,ultra thick] (M2.north) ++(0,0.2cm) -- +(0,1cm);

\end{tikzpicture}
    \caption{Active suspension system.}
    \label{figasc}
\end{figure}

The state vector denoted by  $\bm x = \begin{bmatrix} x_1 & x_2 & x_3 & x_4\end{bmatrix}^T \in \mathbb{R}^4$, where $x_1$ and $x_2$ are the vertical displacements of the body and wheel from their equilibrium position, respectively. $x_3$ and $x_4$ represent the vertical velocity of the body and wheel, respectively, indicating how fast the body and wheel are moving vertically. The control input $u \in \mathbb{R}$ is the force applied to the system, enabling the vehicle to respond to external forces such as road irregularities, which are modeled by $d \in \mathbb{R}$. By controlling the forces applied to the suspension components, the system can achieve improved ride comfort and handling characteristics. Therefore, the control objective is regulating $x_1$ and $x_2$ to zero displacement. However, large displacements for $x_1$ must be prohibited to avoid excessive motion that could impact the safety of the passengers.

Safety in this problem is defined as avoiding vertical displacements larger than a threshold $D$. We consider $D = 6 \hspace{2pt}cm$, and the CBF is $h(\bm x) = D - x_1$, which has relative degree two, with respect to the active suspension system. The nominal model parameters are $m_1 = 300\hspace{2pt}kg, m_1 = 60\hspace{2pt}kg, k_1 = 16\hspace{2pt}kN/m, k_2 = 190\hspace{2pt}kN/m, b = 1\hspace{2pt}kN.s/m$. The true model is obtained by $m_1 = 675\hspace{2pt}kg, m_1 = 135\hspace{2pt}kg, k_1 = 36\hspace{2pt}kN/m, k_2 = 427.5\hspace{2pt}kN/m, b = 2.25\hspace{2pt}kN.s/m$.
We use linear quadratic regulators (LQR) as the stabilization controller with weighting matrices $Q = 10I$ and $R = 1$, and the vector of coefficients in the HOCBF design is $\bm \gamma = \begin{bmatrix} 41 & 395\end{bmatrix}^T$. The system is started from its equilibrium point $\bm{x}_0 = \begin{bmatrix} 0 & 0 & 0 &0 \end{bmatrix}^T$.

We used the same data collecting approach, which collected $174$ data points in total. The same kernel structure as the previous example is selected. Figure \ref{fig3} shows the states $x_1$ and $x_3$ and the control input, together with the road disturbance profile. It highlights that the nominal QP-HOCBF controller violates the threshold, while the proposed GP-based SOCP controller successfully recovers the performance of the true system. 

\begin{figure}[t]
    \centering
    \includegraphics[width=8.5cm]{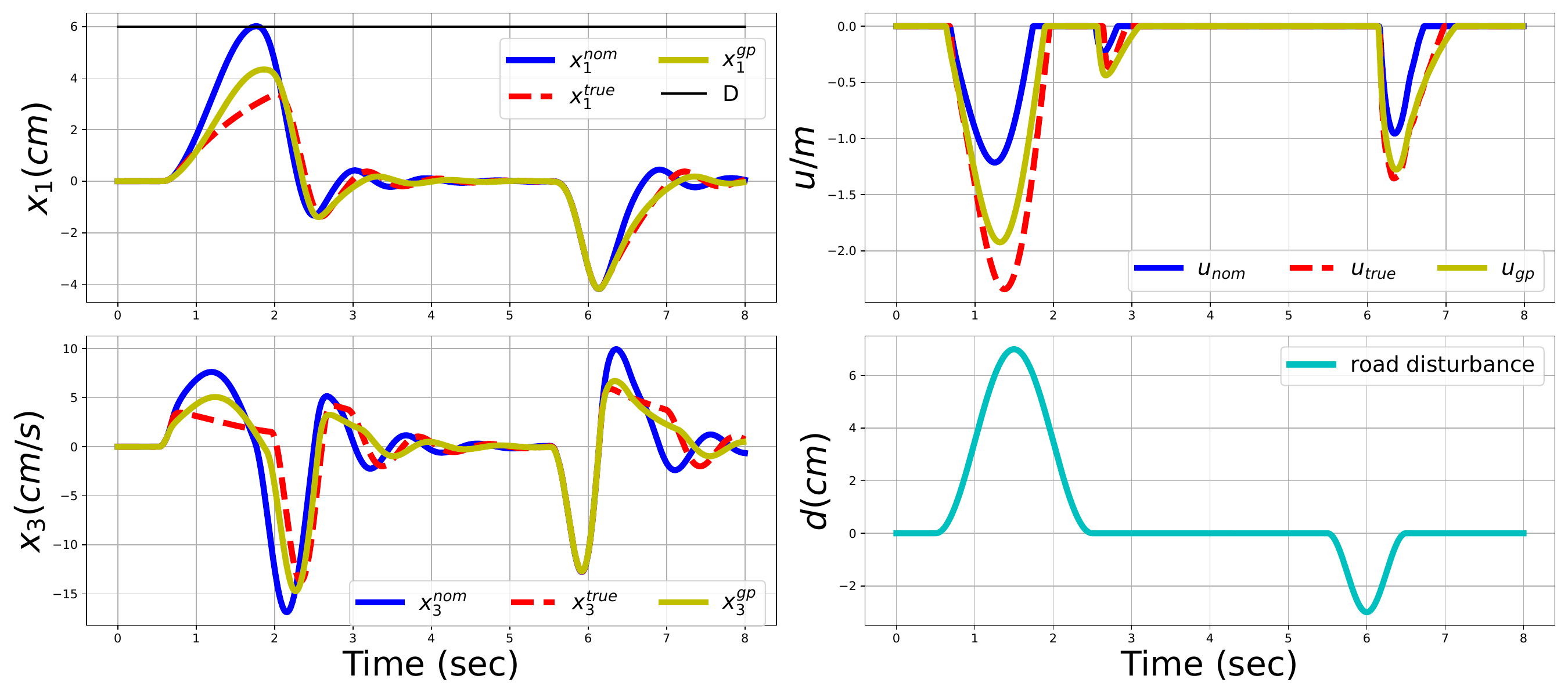}
    \caption{Comparison of the GP-based SOCP-HOCBF (yellow), nominal QP-HOCBF (blue), and true QP-HOCBF (dashed red). States $x_1$ and $x_3$ of the active suspension system (left column). Control input and the road disturbance profile (right column).}
    \label{fig3}
\end{figure}

\section{Conclusion}

In this paper, A GP framework is developed to approximate the effect of model uncertainty on the high-order safety certificate. By choosing linear class $\mathcal{K}$ functions in HOCBF design, the residual term in the high-order safety certificate is fully characterized. Then, the resulting constrained optimization problem with the uncertainty-aware chance constraint is converted into an SOCP which can be solved in real-time. Finally, the feasibility of the SOCP is addressed.
Future works will focus on eliminating the assumption of global relative degree for the true system.
\bibliographystyle{plain}
\bibliography{References}

\begin{thebibliography}{10}

\bibitem{aali2022multiple}
Mohammad Aali and Jun Liu.
\newblock Multiple control barrier functions: An application to reactive obstacle avoidance for a multi-steering tractor-trailer system.
\newblock In {\em Proc. of CDC}, pages 6993--6998, 2022.

\bibitem{ames2019control}
Aaron~D Ames, Samuel Coogan, Magnus Egerstedt, Gennaro Notomista, Koushil Sreenath, and Paulo Tabuada.
\newblock Control barrier functions: Theory and applications.
\newblock In {\em Proc. of ECC}, pages 3420--3431, 2019.

\bibitem{blanchini2008set}
Franco Blanchini, Stefano Miani, et~al.
\newblock {\em Set-theoretic methods in control}, volume~78.
\newblock Springer, 2008.

\bibitem{castaneda2021gaussian}
Fernando Castaneda, Jason~J Choi, Bike Zhang, Claire~J Tomlin, and Koushil Sreenath.
\newblock Gaussian process-based min-norm stabilizing controller for control-affine systems with uncertain input effects and dynamics.
\newblock In {\em Proc. of ACC}, pages 3683--3690, 2021.

\bibitem{castaneda2021pointwise}
Fernando Casta{\~n}eda, Jason~J Choi, Bike Zhang, Claire~J Tomlin, and Koushil Sreenath.
\newblock Pointwise feasibility of gaussian process-based safety-critical control under model uncertainty.
\newblock In {\em Proc. of CDC}, pages 6762--6769, 2021.

\bibitem{Diamond2016}
Steven Diamond and Stephen Boyd.
\newblock Cvxpy: A python-embedded modeling language for convex optimization.
\newblock {\em The Journal of Machine Learning Research}, 17(1):2909--2913, 2016.

\bibitem{fan2020bayesian}
David~D Fan, Jennifer Nguyen, Rohan Thakker, Nikhilesh Alatur, Ali-akbar Agha-mohammadi, and Evangelos~A Theodorou.
\newblock Bayesian learning-based adaptive control for safety critical systems.
\newblock In {\em Proc. of ICRA}, pages 4093--4099, 2020.

\bibitem{jagtap2020control}
Pushpak Jagtap, George~J Pappas, and Majid Zamani.
\newblock Control barrier functions for unknown nonlinear systems using gaussian processes.
\newblock In {\em Proc. of CDC}, pages 3699--3704, 2020.

\bibitem{khalil2002nonlinear}
Hassan~K Khalil.
\newblock {\em Nonlinear Systems}.
\newblock Prentice Hall, 2002.

\bibitem{khojasteh2020probabilistic}
Mohammad~Javad Khojasteh, Vikas Dhiman, Massimo Franceschetti, and Nikolay Atanasov.
\newblock Probabilistic safety constraints for learned high relative degree system dynamics.
\newblock In {\em Proc. of L4DC}, pages 781--792. PMLR, 2020.

\bibitem{nguyen2016exponential}
Quan Nguyen and Koushil Sreenath.
\newblock Exponential control barrier functions for enforcing high relative-degree safety-critical constraints.
\newblock In {\em Proc. of ACC}, pages 322--328, 2016.

\bibitem{srinivas2009gaussian}
Niranjan Srinivas, Andreas Krause, Sham~M Kakade, and Matthias Seeger.
\newblock Gaussian process optimization in the bandit setting: No regret and experimental design.
\newblock {\em arXiv preprint arXiv:0912.3995}, 2009.

\bibitem{tan2021high}
Xiao Tan, Wenceslao~Shaw Cortez, and Dimos~V Dimarogonas.
\newblock High-order barrier functions: Robustness, safety, and performance-critical control.
\newblock {\em IEEE Transactions on Automatic Control}, 67(6):3021--3028, 2021.

\bibitem{taylor2020learning}
Andrew Taylor, Andrew Singletary, Yisong Yue, and Aaron Ames.
\newblock Learning for safety-critical control with control barrier functions.
\newblock In {\em Proc. of L4DC}, pages 708--717, 2020.

\bibitem{wang2021learning}
Chuanzheng Wang, Yiming Meng, Yinan Li, Stephen~L Smith, and Jun Liu.
\newblock Learning control barrier functions with high relative degree for safety-critical control.
\newblock In {\em Proc. of ECC}, pages 1459--1464, 2021.

\bibitem{wang2018safe}
Li~Wang, Evangelos~A Theodorou, and Magnus Egerstedt.
\newblock Safe learning of quadrotor dynamics using barrier certificates.
\newblock In {\em Proc. of ICRA}, pages 2460--2465, 2018.

\bibitem{williams2006gaussian}
Christopher~KI Williams and Carl~Edward Rasmussen.
\newblock {\em Gaussian Processes for Machine Learning}, volume~2.
\newblock MIT Press, 2006.

\bibitem{xiao2021high}
Wei Xiao and Calin Belta.
\newblock High order control barrier functions.
\newblock {\em IEEE Transactions on Automatic Control}, 67(7):3655--3662, 2022.

\end{thebibliography}
\end{document}